\documentclass{amsproc}
\usepackage{amsfonts}
\usepackage{amssymb, amsmath, amsthm}
\usepackage[dvips]{graphics}
\usepackage{epsfig}
\usepackage{euscript}
\usepackage{color}

\usepackage{fancyvrb}

\usepackage[pdftex]{hyperref}

\newtheorem{theorem}{Theorem}[section]
\newtheorem{lemma}[theorem]{Lemma}
\newtheorem{prop}[theorem]{Proposition}

\theoremstyle{definition}

\theoremstyle{remark}

\numberwithin{equation}{section}

\newcommand{\caA}{{\mathcal A}}
\newcommand{\caB}{{\mathcal B}}

\newcommand{\caF}{{\mathcal F}}

\newcommand{\caH}{{\mathcal H}}
\newcommand{\caI}{{\mathcal I}}
\newcommand{\caJ}{{\mathcal J}}

\newcommand{\caO}{{\mathcal O}}

\newcommand{\bbN}{{\mathbb N}}

\newcommand{\bbR}{{\mathbb R}}
\newcommand{\bbS}{{\mathbb S}}

\newcommand{\ie}{{\it i.e.\/} }
\newcommand{\eg}{{\it e.g.\/} }

\newcommand{\iu}{\mathrm{i}}

\newcommand{\str}{^{*}}
\newcommand{\ep}[1]{\mathrm{e}^{#1}}

\newcommand{\Tr}{\mathrm{Tr}}

\newcommand{\be}{\begin{equation}}
\newcommand{\ee}{\end{equation}}
\newcommand{\bea}{\begin{eqnarray}}
\newcommand{\eea}{\end{eqnarray}}
\newcommand{\beann}{\begin{eqnarray*}}
\newcommand{\eeann}{\end{eqnarray*}}

\begin{document}

\title{The adiabatic theorem in a quantum many-body setting}

\author{Sven Bachmann}
\address{Department of Mathematics, The University of British Columbia, Vancouver, BC V6T 1Z2, Canada}
\email{sbach@math.ubc.ca}

\author{Wojciech De Roeck}
\address{Instituut Theoretische Fysica, KULeuven, 3001 Leuven, Belgium}
\email{wojciech.deroeck@kuleuven.be}

\author{Martin Fraas}
\address{Department of Mathematics, Virginia Tech, Blacksburg, VA 24061-0123, USA}
\email{fraas@vt.edu}

\date{\today }

\begin{abstract}
In these lecture notes, we review the adiabatic theorem in quantum mechanics, focusing on a recent extension to many-body systems. The role of locality is emphasized and the relation to the quasi-adiabatic flow discussed. An important application of these results to linear response theory is also reviewed.
\end{abstract}

\maketitle

\section{The adiabatic principle}\label{sec:intro}

The story of the adiabatic theorem in quantum theory is long and rich and these notes are no attempt at an overview of the subject. Let us only point out that the first proof of an adiabatic theorem under the assumption of a spectral gap is due to Fock and Born~\cite{BornFock} rather shortly after quantum theory became a somewhat mature subject, and that it is Kato~\cite{Kato50} who took the matter to a modern mathematical theory. The parallel transport that shall be extensively discussed here was in particular introduced in his famous paper. For a simple proof in a general gapped setting, we further refer to~\cite{NenciuBasic}. 

Among the many refinements and extensions of the basic theorem, the adiabatic expansion introduced in~\cite{Garrido} and reinvigorated in~\cite{BerrySuper} will be particularly relevant for the purpose of the present note. In fact, the proof in the many-body setting that we shall discuss here owes much to the idea of an adiabatic expansion, with the additional ingredient that a special attention is given to `locality' at each step of the iterative procedure.

The adiabatic principle we discuss here applies to the linear dynamics in a Banach space governed by a slowly varying generator:
\begin{equation}\label{Banach}
\epsilon \dot \psi_\epsilon(s) = L_{\sigma_s}\psi_\epsilon(s),\qquad \psi_\epsilon(0) = \psi_0,
\end{equation}
where $s\mapsto\sigma_s$ is a given dependence of a parameter $\sigma$ on time. The starting point is the assumption that for each $\sigma$ the space decomposes into the decaying and the stationary part of the `frozen' dynamics by a projection $P_\sigma$.  Let $\phi_\sigma(t)$ be the solution of the dynamics generated by  $L_\sigma$ with $\sigma$ fixed, namely $\frac{d}{dt}\phi_\sigma(t) = L_{\sigma}\phi_\sigma(t)$; one assumes that $P_\sigma \phi_\sigma(t)$ is constant and  $(1- P_\sigma) \phi_\sigma(t)\rightharpoonup 0$ as $t\to\infty$ in an appropriate sense. 
The adiabatic principle states that if the initial condition $\psi_0$ of the driven dynamics is given by the fixpoint of $L_{\sigma_0}$, namely $(1-P_{\sigma_0})\psi_\epsilon(0) = 0$, then the solution of~(\ref{Banach}) remains close to the instantaneous fixpoint
\begin{equation*}
\left\Vert (1-P_{\sigma_s})\psi_\epsilon(s)\right\Vert = o(1)\qquad (\epsilon\to0)
\end{equation*}
for all $s$ in a compact interval. 

For concreteness, let us consider the example of a smooth family of bounded Hamiltonians $\sigma\mapsto H_\sigma$ on a Hilbert space $\caH$, where $\sigma$ is a parameter such as an external (\ie not dynamical) magnetic field. Assume that $H_\sigma$ has a unique ground state $\varphi_\sigma$ for all $\sigma\in[0,1]$ and that the ground state energy is uniformly isolated from the rest of the spectrum. The driven dynamics is given by letting $\sigma = \epsilon t$ in Schr\"odinger's equation $\iu\frac{d}{dt}\psi(t) = H_{\epsilon t}\psi(t)$ with initial condition $\psi(0) = \varphi_0$. The `slow driving' assumption is implemented by rescaling the time variable $s = \epsilon t$ where the driving rate $\epsilon$ is very small, yielding 
\begin{equation*}
\iu \epsilon \frac{d}{ds}\psi_\epsilon (s) =  H_s \psi_\epsilon (s),\qquad \psi_\epsilon(0) = \varphi_0,
\end{equation*}
where $\psi_\epsilon(s) = \psi(s/\epsilon)$, see~(\ref{Banach}). The adiabatic theorem in the present gapped setting then states that 
\begin{equation*}
\left\Vert \psi_\epsilon (s) - \langle \varphi_s,\psi_\epsilon (s)\rangle \varphi_s \right\Vert \leq C\epsilon,
\end{equation*}
where the constant $C$ naturally depends on the gap. In other words, the dynamically evolved $\psi_\epsilon (s)$ is close to the instantaneous ground state $\ep{\iu\theta_\epsilon(s)}\varphi_s$, up to a possible phase that will be discussed further below.

Note that the physical time involved here is
\begin{equation*}
t\sim\epsilon^{-1}\gg1.
\end{equation*}
The need to control the dynamics over such very long time scales is at the heart of the technical difficulties encountered in making the adiabatic principle into a theorem. 

A concise treatment of adiabatic theory covering both gapped and gapless cases as well as Lindblad evolution and Markov processes can be found in \cite{AFGG}. It however treats the adiabatic principle in the above Banach space setting which, as we shall now explain, is ill-suited for many-body applications.

\section{Many-body issues}\label{sec:MB}

Let us consider the concrete and simple case of a chain of $2L+1$ non-interacting spin-$\frac{1}{2}$'s. The driving is given by a uniform magnetic field whose intensity is constant but whose orientation slowly varies:
\begin{equation*}
H_{h_s,L} = -\sum_{x=-L}^L h_s\cdot \sigma_x,
\end{equation*}
where $s\mapsto h_s\in\bbS^2$, and $\sigma_x$ is the vector of three Pauli matrices at site $x$. Note that although the time evolution is not a true many-body dynamics since the state remains a product for all times, it does exhibit the typical large volume problem. The adiabatic evolution of each individual spin is simple: if it starts aligned with the magnetic field, then under slow driving, the spin will slowly follow the orientation of $h_s$, up to an error which is small in the driving rate.  However, although the state of each spin $\psi_{\epsilon,x}(s)$ is indeed close to the instantaneous ground state $\varphi_{h_s,x}$, say $\vert \langle\psi_{\epsilon,x}(s),\varphi_{h_s,x}\rangle\vert \geq C(1-\epsilon)$, the many-body state is almost orthogonal to the many-body ground state as soon as the number of spins is large as compared to the driving rate:
\begin{equation*}
\langle \otimes_{x=-L}^L\psi_{\epsilon,x}(s), \otimes_{x=-L}^L \varphi_{h_s,x} \rangle = \prod_{x=-L}^L\langle \psi_{\epsilon,x}(s), \varphi_{h_s,x}\rangle \sim (1-\epsilon)^{2L+1}.
\end{equation*}
This phenomenon is well-known as the `orthogonality catastrophe' that plagues the analysis of systems in the thermodynamic limit, see~\cite{lychkovskiy2017time} for a discussion in the adiabatic context.

The example clearly shows that the adiabatic principle as formulated so far for the state just cannot hold. But when it comes to many-body systems, local observables tend to be better-behaved objects than Hilbert space vectors, simply because they probe only a finite part of the state. And indeed, we will show that the expectation value of any local observable in the dynamically evolved state is close to its ground state expectation value. Not surprisingly, the error bound depends on the support of the observable. For the reader looking for analogies in analysis, this is really as simple as going from the $l^1$-topology to $l^\infty$, or more closely from the uniform operator topology to the weak operator topology.

Although this is compelling and indeed true, the situation is that of a child trying to stop a stream by putting a rock in the middle of it: the water leaks around it. The real many-body issue does not appear in the above example because there are no interactions there and hence no dynamical propagation. Luckily for the child, there is also clay to support the rock and ensure that the leaks can be controlled for times long enough that a peaceful puddle will be created behind the dam.

In order to clarify further the issue at stake, let us provide a simple derivation of the adiabatic theorem in quantum mechanics in the presence of a spectral gap. Let $P_s$ be the projection onto the ground state space of $H_s$ and let $\rho_\epsilon(s)$ be its adiabatic evolution starting from $P_0$, namely
\begin{equation*}
\epsilon\dot\rho_\epsilon(s) = -\iu[H_s,\rho_\epsilon(s)],\qquad \rho_\epsilon(0)=P_0.
\end{equation*}
We assume that the driving is compactly supported so that $\dot P\vert_{s=0} = \dot P\vert_{s=1} = 0$. We denote
\begin{equation*}
L_s(\cdot) = -\iu[H_s,\cdot]
\end{equation*}
and let $\sigma_\epsilon^{s,s'}(\cdot)$ be the flow generated by $L_s$, namely $\sigma_\epsilon^{s,0}(P_0) = \rho_\epsilon(s)$. Since $L_s(P_s) = 0$, the diabatic error $r_\epsilon(s) = \rho_\epsilon(s) - P_s$ is a solution of the initial value problem
\begin{equation*}
\left(\epsilon\frac{d}{ds} - L_s\right)r_\epsilon(s) = -\epsilon \dot P_s,\qquad r_\epsilon(0) = 0.
\end{equation*}
Hence it can be expressed using Duhamel's principle as
\begin{equation*}
r_\epsilon(s) = -\int_0^s\sigma_\epsilon^{s,s'}(\dot P_{s'})ds'.
\end{equation*}
In order to `integrate by parts', we note that $\epsilon \frac{d}{ds'}\sigma_\epsilon^{s,s'}(Q) = - \sigma_\epsilon^{s,s'}(L_{s'}(Q))$ for any $Q$ to conclude that
\begin{equation}\label{rest}
\rho_\epsilon(1) - P_1 = - \epsilon \int_0^1\sigma_\epsilon^{1,s}\left(\frac{d}{ds}L_s^{-1}(\dot P_s)\right)ds.
\end{equation}
Since $\dot{P}_s$ is off-diagonal, see~(\ref{Pod}), the inverse of $L_s$ is well defined. This concludes this sketch of proof of the adiabatic theorem under sufficient smoothness assumptions on $s\mapsto H_s$, since the Schr\"odinger flow $\sigma_\epsilon^{1,s}$ is norm preserving. 

Let us first understand how~(\ref{rest}) connects to the infrared catastrophe. If $P_s$ is the projection onto a spectral patch $\Sigma_s\subset[a,b]$ that is isolated from the rest of the spectrum by a gap $\gamma>0$, then Riesz' formula 
\begin{equation*}
P_s = -\frac{1}{2\pi\iu}\int_\Gamma (H_s - z)^{-1} dz
\end{equation*}
and the spectral theorem imply that
\begin{equation*}
P_s\dot P_s(1-P_s) = \frac{1}{2\pi\iu}\int_\Gamma\Big(\int_a^b \frac{1}{\lambda - z} dP_s(\lambda)\Big)\dot H_s \Big( \int_{\bbR\setminus[a,b]}\frac{1}{\mu - z} dP_s(\mu)\Big) dz.
\end{equation*}
With $\Gamma$ being a simple closed contour winding once around the patch $\Sigma_s$, the only pole contributing to the integral is at $\lambda$ so that
\begin{equation*}
P_s\dot P_s(1-P_s) = \int_a^b\int_{\bbR\setminus[a,b]} \frac{1}{\lambda - \mu} dP_s(\lambda)\dot H_sdP_s(\mu).
\end{equation*}
By the gap condition, $\vert \lambda - \mu\vert\geq \frac{\gamma}{2}$ and hence
\begin{equation*}
\Vert \dot P_s \Vert \leq C \frac{\Vert \dot H_s\Vert}{\gamma},
\end{equation*}
since
\begin{equation}\label{Pod}
\dot P_s = P_s\dot P_s(1-P_s) + (1-P_s)\dot P_sP_s
\end{equation}
because $P_s = P_s^2$. If the change in the Hamiltonian is truly extensive, then $\Vert \dot H_s\Vert$ is of order of the volume $V$ (or equivalently of the number of particles) and~(\ref{rest}) is meaningful only in the `few-body' regime $\epsilon V \ll 1$.

This being a norm estimate, the collapse of the argument may be circumvented by considering a weaker topology, namely expectation values of local operators, as discussed above. However, expectation values are not invariant under unitary conjugation. This invariance was however crucial in the norm estimate of~(\ref{rest}) because it allowed us to ignore the long-time dynamics $\sigma_\epsilon^{1,s}$. In fact, for any local observable $A$,
\begin{equation*}
\Tr\left(A\sigma_\epsilon^{1,s}(\frac{d}{ds}L_s^{-1}(\dot P_s))\right) = \Tr\left(\tau_\epsilon^{1,s} (A) \frac{d}{ds} (L_s^{-1}(\dot P_s))\right)
\end{equation*}
 and the norm of the operators inside the trace must be expected to grow with the support of the Heisenberg evolved $\tau_\epsilon^{1,s} (A)$. By the Lieb-Robinson bound~\cite{Lieb:1972ts,Nachtergaele:2006fd}, this dynamics is in general ballistic: the linear size of the support is proportional to time. The adiabatic dynamics running for long times of order $\epsilon^{-1}$, the bound to be expected in a generic many-body situation is of order $\epsilon^{-d}$ in $d$ spatial dimensions. This  yields a naive diabatic error of order $\epsilon^{1-d}$, see~(\ref{rest}) which is useful only if $d=0$, the one-body limit of many-body quantum dynamics.

In bypassing the infinite volume infrared problem and concentrating on the long time issue, it is clear how to progress: Under sufficient smoothness of the Hamiltonian, the linear estimate~(\ref{rest}) can indeed be improved to a bound of order $\epsilon^n$ and even exponential if the Hamiltonian is infinitely differentiable, see~\cite{Nenciu}.
Hence, the challenge is to construct an adiabatic expansion inspired by these works to order $n\geq d+1$, which moreover respects the locality of the dynamics.

In the rest of these notes, we provide the arguments necessary to carry out the program just sketched. Although we skip some of the more tedious technical estimates, we provide all the necessary arguments, and refer the interested reader to the original~\cite{OurMathAdiabatic} for details in the setting of quantum spin systems. We also note that the same ideas, but with the necessary technical modifications, have successfully been applied to lattice fermions in~\cite{TeufelAd} with similar results.

\section{Locality and parallel transport}\label{sec:Locality}

The argument leading to~(\ref{rest}) uses $s\mapsto P_s$ as a given map without further structure. As Kato first noticed, it is useful in the context of adiabatic theory to cast it as a flow
\begin{equation}\label{Flow}
P_s = U^{\mathrm{A}}(s) P_0 U^{\mathrm{A}}(s)\str,\qquad U^{\mathrm{A}}(0) = 1,
\end{equation}
and to compare $U^{\mathrm{A}}(s)$ with the Schr\"odinger propagator $U_\epsilon(s,0)$, instead of comparing directly the projectors $P_s$ and $\rho_\epsilon(s)$.

Let $\Omega_s$ be a unit vector in the range of $P_s$, namely $P_s\Omega_s = \Omega_s$. Then
\begin{equation}\label{transport}
(1-P_s)\dot\Omega_s = \dot P_s\Omega_s 
\end{equation}
which shows that the constraint of having a smooth family of states belonging to $\mathrm{Ran}(P_s)$ at each $s$ actually leaves the motion of $\Omega_s$ within $\mathrm{Ran}(P_s)$ undetermined. Note that this remark remains relevant even in the case of a non-degenerate eigenvalue as the phase of $\Omega_s$ is undetermined. A natural choice is the following parallel transport condition
\begin{equation}\label{PT}
P_s\dot\Omega_s = 0,
\end{equation}
which imposes that there is no motion at all within $\mathrm{Ran}(P_s)$ as $s$ changes. Adding this and (\ref{transport}) yields $\dot\Omega_s = \dot P_s \Omega_s$ and with a little more algebra using~(\ref{Pod}),
\begin{equation*}
\dot\Omega_s = (1-P_s)\dot P_s P_s \Omega_s = [\dot P_s, P_s]\Omega_s.
\end{equation*}
This is Kato's choice of a generator for the flow, namely
\begin{equation*}
\Omega_s = U^{\mathrm{K}}(s)\Omega_0,
\end{equation*}
where
\begin{equation*}
\iu \dot U^{\mathrm{K}}(s) = \iu[\dot P_s, P_s] U^{\mathrm{K}}(s).
\end{equation*}

We emphasize at this point that there are really two choices that have been made, namely first the choice of map $s\mapsto\Omega_s$ within $\mathrm{Ran}(P_s)$, and second the choice of unitary implementing this map. It is also to be noted that Kato's generator $\iu[\dot P_s, P_s]$ has a priori no reason to be a local Hamiltonian and the corresponding dynamics has no reason to satisfy a Lieb-Robinson bound. Since it is however a fact that the driven Schr\"odinger propagator is local, it is immediately clear that $U^{\text{K}}(s)$ may not be a good choice to approximate the Schr\"odinger propagator in the many-body context.

The differential version of~(\ref{Flow}) reads
\begin{equation}\label{Generator}
\dot P_s = -\iu [G_s^{\text{A}}, P_s],
\end{equation}
where
\begin{equation*}
G_s^{\text{A}} := \iu \dot U^{\text{A}}(s)U^{\text{A}}(s)\str
\end{equation*}
is the generator of $U^{\text{A}}(s)$. 
Let now $W\in L^1(\bbR;\bbR)$ and 
\begin{equation}\label{GH}
G_s^{\text{H}} := \int_\bbR W(t)\ep{\iu t H_s}\dot H_s\ep{-\iu t H_s} d t,
\end{equation}
which is a norm-convergent integral. Since $[H_s,P_s] = 0$, we see that $[\dot H_s ,P_s ] = [\dot P_s,H_s]$ and hence
\begin{align}
-\iu[G_s^{\text{H}},P_s] &= -\iu\int_\bbR W(t)\ep{\iu t H_s}[\dot P_s,H_s]\ep{-\iu t H_s} d t \label{[GP]}\\
&=-\iu\sqrt{2\pi}\int_{\bbR^2} \widehat{W}(\mu-\lambda) (\mu-\lambda) dP_s(\lambda)\dot P_s dP_s(\mu),\nonumber
\end{align}
where $\widehat W$ denotes the Fourier transform of $W$. Since $\dot P_s$ is off-diagonal, see~(\ref{Pod}), the integrand is in fact supported only on the set
\begin{equation*}
\{(\mu,\lambda)\in \mathrm{Spec}(H_s)\times \mathrm{Spec}(H_s): \vert \mu - \lambda\vert \geq \gamma_s\}
\end{equation*}
where $\gamma_s>0$ is the spectral gap of $H_s$. One concludes that $G_s^{\text{H}}$ solves~(\ref{Generator}) for all $s$ if and only if
\begin{equation}\label{FourierW}
\widehat{W}(\xi) = \frac{\iu}{\sqrt{2\pi}\xi},\qquad \vert\xi\vert\geq \inf_s\gamma_s.
\end{equation}
Note that the non-integrable decay of $\widehat W$ at infinity hints at a discontinuity of $W$. This could be seen informally by observing that~(\ref{[GP]}) can also be written as 
\begin{equation*}
-\iu[G_s^{\text{H}},P_s] = \int_\bbR W(t)\frac{d}{dt}\ep{\iu t H_s}\dot P_s\ep{-\iu t H_s} d t
\end{equation*}
which yields $\dot P_s$ upon integration by parts if $W' = - \delta$.

Since $H_s$ is a sum of local terms (see~\cite{PieterQSS} as well as Appendix~\ref{appendix:Local} for details on the setting of quantum lattice systems),
\begin{equation}\label{Hamiltonian}
H_s = \sum_{X} \Phi_s(X),
\end{equation}
where $\Phi_s(X)$ acts only on the subset $X$, so is $\dot H_s$. By the Lieb-Robinson bound, $\ep{\iu t H_s}\dot  \Phi_s(X) \ep{-\iu t H_s}$ remains an almost local operator for times of order $1$. Hence, for $G_s^{\text{H}}$ defined in~(\ref{GH}) to be a sum of local terms, it suffices that $W$ is a function of fast decay. In turn, if this holds, then the dynamics generated by $G_s^{\text{H}}$ satisfies a Lieb-Robinson bound.

Finally, if the spectral patch $\Sigma_s$ is in fact a single point, then
\begin{equation*}
P_s G_s^{\text{H}} P_s = \sqrt{2\pi} \, \widehat W(0) P_s \dot H_s P_s,
\end{equation*}
and $\widehat W(0) = 0$ implies that the vector $U^{\mathrm{H}}(s)\Omega_0\in\mathrm{Ran}(P_s)$ satisfies the parallel transport condition~(\ref{PT}) since
\begin{equation*}
P_s\frac{d}{ds}U^{\mathrm{H}}(s)\Omega_0 = -\iu P_s G_s^{\text{H}} U^{\mathrm{H}}(s)\Omega_0.
\end{equation*}
The condition $\widehat W(0) = 0$ is satisfied in particular if $W$ is an odd function.

Let us add a few words about $W$. As done in~\cite{automorphic}, $W$ can be constructed as the antiderivative of $w$ whose Fourier transform vanishes outside a finite interval. If $w(t) = \caO(\exp(-\vert t\vert \xi))$ for a $\xi>0$ at $t\to\pm\infty$, then its Fourier transform is a holomorphic function, which must be identically zero if it vanishes outside an interval. While exponential decay is therefore impossible to achieve, almost any subexponential decay is, see~\cite{InghamFT}. Precisely, if $\xi$ is a continuous positive function such that 
\begin{equation*}
\int^\infty\frac{\xi(t)}{t}dt<\infty,
\end{equation*}
then there exists a compactly supported function $\chi$ such that 
\begin{equation*}
\check \chi(t) = \caO(\ep{-\vert t\vert\xi(\vert t\vert)}).
\end{equation*}

The generator $G_s^{\text{H}}$ was first proposed in a slightly different fashion by Hastings in~\cite{Hastings:2004go} and further developed in~\cite{HastingsWen}. The present form including an explicit function satisfying all three conditions above with a subexponential decay at infinity was introduced in~\cite{automorphic}.

Summarizing, the flow generated by the local Hamiltonian $G_s^{\text{H}}$ with $\widehat W(0) = 0$ is exactly the same as Kato's flow within $\mathrm{Ran}(P_s)$, but the two differ otherwise on $\caH$ in such a way that Hastings' flow is local in the sense that it satisfies a Lieb-Robinson bound. We shall follow the general agreement and call $A\mapsto U^{\text{H}}(s)\str A U^{\text{H}}(s)$ the `quasi-adiabatic flow'.

\section{Local adiabatic expansion}\label{sec:Expansion}

With this local version of $U^{\mathrm{A}}(s)$ in hand, we turn to the heart of the argument. The ideas presented in this section, which collectively go under the name of superadiabatic expansion, can be traced back to Garrido~\cite{Garrido}, and developed by Berry~\cite{Berry90} and later~\cite{Nenciu}. We also mention~\cite{JoyePfisterSuper} which explains the iteration scheme in a language that is very close to the one used here. The novelty in the many-body setting is to carry out the construction so as to keep locality at all orders. The result thereof is the many-body adiabatic theorem of~\cite{OurMathAdiabatic}.

The theorem below is set on an arbitrary but finite subset $\Lambda$. As we discussed in Section~\ref{sec:MB}, the key however is to obtain bounds that are independent of the volume $\vert \Lambda\vert$. If not otherwise specified, all bounds are uniform in $\Lambda$.

Let $\psi_\epsilon(s) = U_\epsilon(s,0)\Omega_0$ be the solution of the driven Schr\"odinger equation for a Hamiltonian of the form~(\ref{Hamiltonian}) defined on finite subsets $\Lambda$ of an infinite lattice, and let $\Omega_s$ be the parallel transported ground state, see~(\ref{PT}). Let $d$ be the spatial dimension of the lattice, defined by the rate of growth of balls. The Hamiltonian is of the form~(\ref{Hamiltonian}) where the interaction terms $\Phi_s(X)$, together with their derivatives, satisfy sufficient decay conditions in the size of $X$, see Assumption~2.2 in~\cite{OurMathAdiabatic}. In fact, the original article just cited is mistaken on the order of differentiability of the Hamiltonian that is needed.
\begin{theorem}\label{thm:MBA}
Let $H\in C^{d+2}([0,1])$ be a local quantum spin Hamiltonian with ground state energy $E_s$ isolated from the rest of the spectrum, uniformly in $s\in[0,1]$ and $\Lambda$. Assume that all derivatives of $H$ vanish at $s=0$. Then for any $A\in\caA^{\mathrm{loc}}$, 
\begin{equation*}
\left\vert \langle\psi_\epsilon(1),A\psi_\epsilon(1)\rangle 
- \langle\Omega_1,A\Omega_1\rangle\right\vert \leq C_1 \epsilon. 
\end{equation*}
If $H\in C^{d+1+m}([0,1])$ for $m\geq 1$ and its derivative is compactly supported in $(0,1)$, then
\begin{equation*}
\left\vert \langle\psi_\epsilon(1),A\psi_\epsilon(1)\rangle 
- \langle\Omega_1,A\Omega_1\rangle\right\vert \leq C_m \epsilon^m. 
\end{equation*}
All constants $C_j$ are independent of $\Lambda$.
\end{theorem}

We now navigate through the proof of the theorem. Instead of comparing the adiabatically evolved $\rho_\epsilon(s)$ to the instantaneous stationary state $P_s$, we aim to construct recursively a dressing transformation $V_\epsilon^{(n)}(s)$ for $n=1,2,\ldots$ such that
\begin{equation}\label{Pin}
\Pi_\epsilon^{(n)}(s) = V_\epsilon^{(n)}(s) P_s V_\epsilon^{(n)}(s)\str
\end{equation}
is close to $\rho_\epsilon(s)$, while its range remains $\epsilon$-close to that of $P_s$. The second requirement is immediately satisfied by choosing the Ansatz
\begin{equation}\label{Vs}
V_\epsilon^{(n)}(s) = \ep{\iu S^{(n)}_\epsilon(s)}
\end{equation}
where 
\begin{equation*}
S^{(n)}_\epsilon(s) = \sum_{p = 1}^n \epsilon^p A_p(s),
\end{equation*}
and $A_p$'s are all local Hamiltonians to be determined. In order to satisfy the first requirement, we start by noting that
\begin{multline*}
\iu\epsilon\dot \Pi_\epsilon^{(n)}(s)
= [H_s, \Pi_\epsilon^{(n)}(s)] \\
 + V_\epsilon^{(n)}(s)\left[\iu \epsilon V_\epsilon^{(n)}(s)\str \dot V_\epsilon^{(n)}(s) + \epsilon G_s^{\mathrm{H}} + (H_s - V_\epsilon^{(n)}(s)\str H_s V_\epsilon^{(n)}(s)),P_s\right]V_\epsilon^{(n)}(s)\str.
\end{multline*}
We have separated the zeroth order in $\epsilon$ of the right hand side to emphasize the proximity of $\Pi_\epsilon^{(n)}(s)$ to $\rho_\epsilon(s)$. We now choose $A_p$'s so as to cancel the second term, order by order in~$\epsilon$. To order $p=1$, we obtain the following linear equation for $A_1(s)$:
\begin{equation*}
\left[G_s^{\mathrm{H}} + \iu[A_1(s) , H_s],P_s\right] = 0.
\end{equation*}
By Lemma~\ref{lma:LocalInverse}, the choice
\begin{equation*}
A_1(s) = -\caI(G_s^{\mathrm{H}})
\end{equation*}
is a solution. Since $G_s^{\mathrm{H}}$ is a sum of almost local terms, so is $A_1(s)$ (see the discussion in Appendix~\ref{appendix:L}), thereby closing the first step of the recursion. Moreover, the assumption $\dot H\vert_{s=0} = 0$ implies $G_s^{\mathrm{H}}\vert_{s=0} = 0$ and hence $A_1(0) = 0$.

The procedure is now clear although slightly tedious. Expand
\begin{equation}\label{ExpGen}
\iu \epsilon V_\epsilon^{(n)}(s)\str \dot V_\epsilon^{(n)}(s) 
= -\epsilon\int_0^1\ep{-\iu\mu S^{(n)}_\epsilon(s)}\dot S^{(n)}_\epsilon(s)\ep{\iu\mu S^{(n)}_\epsilon(s)}d\mu = \sum_{p = 2}^{n}\epsilon^p T_p(s) + \tilde T_{n+1}(s)
\end{equation}
as well as 
\begin{equation}\label{ExpH}
H_s - V_\epsilon^{(n)}(s)\str H_s V_\epsilon^{(n)}(s) = \sum_{p = 1}^{n}\epsilon^p K_p(s) + \tilde K_{n+1}(s)
\end{equation}
and choose $A_p$ recursively to cancel out the terms of order $\epsilon^p$. We find for the first steps
\begin{equation*}
K_1(s) = \iu[A_1(s),H_s],\qquad K_2(s) = \iu[A_2(s),H] - \frac{1}{2}[A_1(s),[A_1(s),H_s]],
\end{equation*}
and
\begin{equation*}
T_2(s) = -\dot A_1(s).
\end{equation*}
The higher order terms can be expressed in terms of multi-commutators of the $A_p$'s and their derivatives. It is worthwhile to note that the equation at order $p$ involves $\{A_j,1\leq j\leq p\}$ as well as time derivatives thereof for $j<p$, and that $A_p$ appears only in $K_p$, and in fact through its commutator with the Hamiltonian. As in the case $p=1$, this special structure ensures that Lemma~\ref{lma:LocalInverse} provides a solution for $A_p$ in terms of the lower order potentials. Moreover, the set of operators which are sums of almost local terms, equipped with the sum and the product $[\cdot,\cdot]$ is an algebra, ensuring that all $A_p$'s are local Hamiltonians. It follows that the propagator $V_\epsilon^{(n)}(s)$ satisfies a Lieb-Robinson bound. This in turn implies that the rest terms $\tilde T_{n+1}(s),\tilde K_{n+1}(s)$ are themselves of the same form. 

We conclude that the construction will provide a sequence of local Hamiltonians $\{A_p\}$ and that the recursion must only stop if differentiability is lost. Indeed, the existence of $G^{\mathrm{H}}$ requires $H\in C^1$, see~(\ref{GH}), which in turn yields the existence of $A_1$. We further observe that $A_2$ depends on $\dot A_1$ and hence on $\ddot H$, and recursively $A_p$ exists only if $H\in C^p$. In order for resulting dressed projection $\Pi_\epsilon^{(n)}$ to be differentiable, all $A_j$'s must be differentiable. Hence, the construction is well-defined with a rest of order $\epsilon^{n+1}$ provided $H\in C^{n+1}$. Moreover, since all derivatives of $H$ are assumed to vanish at $s=0$, we conclude that $A_1(0) = \ldots = A_p(0) = 0$. Hence, the dressing transformation is trivial at $s=0$.

Summarizing, we have proved the following proposition (and we refer once again to the Appendix~\ref{appendix:Local} for clarifications of the notion of a local Hamiltonian):
\begin{prop}\label{prop:dressing}
Under the assumptions of Theorem~\ref{thm:MBA} with $H$ being $(n+1)$-times differentiable, there exist local Hamiltonians $A_1,\ldots A_n$ such that $\Pi_\epsilon^{(n)}$ defined by~(\ref{Pin},\ref{Vs}) solves
\begin{equation*}
\iu\epsilon \dot \Pi_\epsilon^{(n)}(s) = [H_s + R_\epsilon^{(n)}(s) ,\Pi_\epsilon^{(n)}(s)],\qquad \Pi_\epsilon^{(n)}(0) = P_0,
\end{equation*}
where $R_\epsilon^{(n)}(s)$ is a local Hamiltonian of order $\epsilon^{n+1}$. 
\end{prop}

The rest of the argument yielding the adiabatic theorem~\ref{thm:MBA} is rather routine using Duhamel's formula and the Lieb-Robinson bound. We provide here a slightly different version than the original~\cite{OurMathAdiabatic} which emphasizes the relation to parallel transport.

Recall that $\psi_\epsilon(s) = U_\epsilon(s,0)\Omega_0$ is the solution of the Schr\"odinger equation, and let $\omega_\epsilon^{(n)}(s) = V_\epsilon^{(n)}(s)\Omega_s$ be the dressed ground state. Then
\begin{equation*}
\psi_\epsilon(s) - \omega_\epsilon^{(n)}(s) 
= -\left. U_\epsilon(s,r)\omega_\epsilon^{(n)}(r)\right\vert_{r=0}^{r=s},
\end{equation*}
By Proposition~\ref{prop:dressing}, $\iu\epsilon \frac{d}{ds}\omega_\epsilon^{(n)}(s) = (H_s + R_\epsilon^{(n)}(s))\omega_\epsilon^{(n)}(s)$, so that
for any $A\in\caA^{\mathrm{loc}}$
\begin{multline}\label{RealvsDressed}
\langle\psi_\epsilon(s),A\psi_\epsilon(s)\rangle - \langle\omega_\epsilon^{(n)}(s),A\omega_\epsilon^{(n)}(s)\rangle \\
=\frac{\iu}{\epsilon}\int_0^s\langle \omega_\epsilon^{(n)}(r),[U_\epsilon(s,r)\str AU_\epsilon(s,r),R_\epsilon^{(n)}(r)]\omega_\epsilon^{(n)}(r)\rangle dr.
\end{multline}
The observable $U_\epsilon(s,r)\str AU_\epsilon(s,r)$ is supported on an $\epsilon^{-1}$-fattening of the support of $A$ so that 
\begin{equation*}
\left\Vert [U_\epsilon(s,r)\str AU_\epsilon(s,r), R_\epsilon^{(n)}(r)] \right\Vert = \caO\left(\epsilon^{(n+1)-d}\right)
\end{equation*}
because $R_\epsilon^{(n)}(r)$ is a local Hamiltonian of order $\epsilon^{n+1}$. Furthermore,
\begin{equation}\label{DressedvsNaked}
\langle\omega_\epsilon^{(n)}(s),A\omega_\epsilon^{(n)}(s)\rangle - \langle\Omega_s,A\Omega_s\rangle = \iu \int_0^1\langle\Omega_s,[\ep{-\iu\mu S_\epsilon^{(n)}(s)} A \ep{\iu\mu S_\epsilon^{(n)}(s)},S_\epsilon^{(n)}(s) ]\Omega_s\rangle d\mu
\end{equation}
which is of order $\epsilon$ because $S_\epsilon^{(n)}(s)$ is a local Hamiltonian of order $\epsilon$ and the dynamics it generates satisfies a Lieb-Robinson bound. Hence, if $H\in C^k$ for $k\geq d+2$, we set $n=d+1$ in~(\ref{RealvsDressed}) and conclude with~(\ref{DressedvsNaked}) that, as claimed in the first part of Theorem~\ref{thm:MBA},
\begin{equation*}
\left\vert \langle\psi_\epsilon(s),A\psi_\epsilon(s)\rangle 
- \langle\Omega_s,A\Omega_s\rangle\right\vert \leq C(A)\epsilon,
\end{equation*}
where $C(A)$ depends on $A$ but it is independent of the volume $\Lambda$. A careful analysis of the error yields $C(A) = c \Vert A\Vert \vert \mathrm{supp}(A)\vert^2$ for an $A$-independent constant $c$.

It remains to prove the second claim under the further assumption that the driving has stopped at $s=1$. If the Hamiltonian has a compactly supported derivative, then as already discussed in the paragraph before Proposition~\ref{prop:dressing} about $s=0$, we have $A_1(1) = \cdots = A_k(1) = 0$ since they all depend locally in time on the derivatives of $H$, see~(\ref{ExpGen},\ref{ExpH}) and the definition~(\ref{I}) of $\caI$. Therefore, $S_\epsilon^{(k)}(1) = 0$ in that case, so that the dressed ground state is just the ground state itself, $\omega_\epsilon^{(k)}(1) = \Omega_1$. The diabatic error arises solely from~(\ref{RealvsDressed}), which yields at $n=d+m$ an improved order $\epsilon^{m}$, namely the second part of Theorem~\ref{thm:MBA}.

We conclude this section with two remarks. First of all, if the Hamiltonian is smooth with compactly supported derivative, then the diabatic error is beyond perturbation theory. It is argued in~\cite{OurPRL} that the error is in fact exponential as it is in the single-body case~\cite{Nenciu}, but it is dimension-dependent:
\begin{equation*}
\left\vert \langle\psi_\epsilon(s),A\psi_\epsilon(s)\rangle 
- \langle\Omega_s,A\Omega_s\rangle\right\vert \leq C(A)\ep{-\frac{c}{\epsilon^{1/d}}}. 
\end{equation*}

Secondly, we described so far the unitary $V_\epsilon^{(n)}(s)$ as a dressing transformation and thought of it as the transformation taking the instantaneous ground state $\Omega_s$ to a properly dynamically adjusted version of it. And indeed, $V_\epsilon^{(n)} = 1$ as soon as the driving stops. Alternatively, one could also see $V_\epsilon^{(n)}(s)$ as being so that the vector $V_\epsilon^{(n)}(s)\str U_\epsilon(s,0)\Omega_0$ is close to $\Omega_s$. In other words, we have engineered an additional `counter-diabatic driving' to bring the slowly driven Schr\"odinger dynamics back to pure parallel transport, see~\cite{CDD,Demirplak2008,Saberi}.

\section{Kubo's Formula}\label{sec:Kubo}

We close these notes with an important application to condensed matter physics, namely the proof of the validity of Kubo's formula for disspationless transport at zero temperature in a many-body setting.

We first recall the general setting of linear response that we will have in mind. We consider an initial unperturbed Hamiltonian $H_{\mathrm i}$ upon which a weak perturbation is adiabatically switched on:
\begin{equation*}
H_t = H_{\mathrm i} + \alpha \ep{\epsilon t} V
\end{equation*}
where $t\in(\infty,0]$ and $0<\alpha\ll 1$. Let
\begin{equation*}
P_{\mathrm i} = P_{\{\alpha = 0,t\}} = P_{\{\alpha,t=-\infty\}}
\end{equation*}
be the ground state projection of $H_{\mathrm{i}}$. The state $\rho_{\epsilon,\alpha}(t)$ solves the initial value problem
\begin{equation*}
\frac{d}{dt}\rho_{\epsilon,\alpha}(t) = -\iu[H_t, \rho_{\epsilon,\alpha}(t)],\qquad \lim_{t\to-\infty}\ep{\iu H_{\mathrm i}t}\rho_{\epsilon,\alpha}(t)\ep{-\iu H_{\mathrm i}t} = P_{\mathrm i}.
\end{equation*}
The equation is usually solved in the interaction picture $\varrho_{\epsilon,\alpha}(t) := \ep{\iu H_{\mathrm i}t}\rho_{\epsilon,\alpha}(t)\ep{-\iu H_{\mathrm i}t}$, namely
\begin{equation*}
\varrho_{\epsilon,\alpha}(t) = \varrho_{\epsilon,\alpha}(t_0) - \iu\alpha\int_{t_0}^t\ep{\epsilon \tau}
\left[\ep{\iu H_{\mathrm i}\tau}V\ep{-\iu H_{\mathrm i}\tau},\varrho_{\epsilon,\alpha}(\tau)\right] d\tau.
\end{equation*}
Letting $t=0,t_0\to-\infty$,
\begin{equation*}
\rho_{\epsilon,\alpha}(0) - P_{\mathrm i} = - \iu\alpha\int_{-\infty}^0\ep{\epsilon \tau}
\ep{\iu H_{\mathrm i}\tau}\left[V,\rho_{\epsilon,\alpha}(\tau)\right]\ep{-\iu H_{\mathrm i}\tau} d\tau
\end{equation*}
and we obtain to first order in $\alpha$
\begin{equation*}
\Tr\big(J(\rho_{\alpha,\epsilon}(0) -  P_{\mathrm i})\big) \sim 
- \iu\alpha\int_{-\infty}^0\ep{\epsilon \tau}
\Tr\big(
P_{\mathrm i}\left[\ep{-\iu H_{\mathrm i}\tau}J\ep{\iu H_{\mathrm i}\tau},V\right]
\big) d\tau
\end{equation*}
for any observable $J$. This is the celebrated Kubo formula for the linear response of $J$ under the driving $V$,
\begin{equation}\label{Kubo}
\chi^{\text{Kubo}}_{J,V} = \lim_{\epsilon\to 0^+} \iu\int_{-\infty}^0\ep{\epsilon \tau}
\Tr\big(
P_{\mathrm i}\left[V,\ep{-\iu H_{\mathrm i}\tau}J\ep{\iu H_{\mathrm i}\tau}\right]
\big) d\tau,
\end{equation}
where the right hand side depends exclusively on the unperturbed quantities.

Let $J\in\caA^{\mathrm{loc}}$ be a local current observable in a many-body setting and $V$ be an extensive local Hamiltonian. By the Lieb-Robinson bound the norm of the commutator $\Vert[V,\ep{-\iu H_{\mathrm i}\tau}J\ep{\iu H_{\mathrm i}\tau}]\Vert$ is expected to be of order $\tau^d$. It follows by scaling that the integral in~(\ref{Kubo}) is of order $\epsilon^{-(d+1)}$ as $\epsilon\to0^+$, questioning the validity of Kubo's formula and of its derivation.

Let us first clarify an important point which has remained conveniently hidden in the algebra above, namely that of the order of the three limits involved in the many-body setting: the thermodynamic limit $\vert \Lambda\vert \to\infty$, the adiabatic limit $\epsilon\to0^+$ and the linear response limit $\alpha\to0^+$. The physically correct definition of the linear response coefficient is
\begin{equation}\label{chi}
\chi_{J,V} :=\lim_{\alpha\to 0^+}\lim_{\epsilon\to 0^+}\lim_{\vert \Lambda\vert \to\infty}\frac{\Tr(J\rho_{\alpha,\epsilon}(0)) - \Tr(J P_{\mathrm i})}{\alpha},
\end{equation}
and the infinite volume being taken first jeopardizes Kubo's argument. Nonetheless, the combination of the quasi-adiabatic flow and the many-body adiabatic theorem allows for a simple proof of the existence of these limits.

For simplicity, we assume here the existence of the thermodynamic limit in the sense that $\Tr(A P_{\mathrm i})\to\varpi_{\mathrm i}(A)$ for any $A\in\caA^{\mathrm{loc}}$ as $\vert \Lambda\vert \to\infty$. Then following theorem not only claims the validity of linear response theory in the sense of the existence of $\chi_{J,V}$, it also provides a new formula for it, and finally claims that Kubo's formula is in fact well-defined even if the thermodynamic limit is taken first. 
\begin{theorem}\label{thm:LR}
Assume that $H_{\mathrm i},V$ are local quantum spin Hamiltonians and that $H_{\mathrm i} + \beta V$ is gapped in a neighbourhood of $\beta = 0$. Then
\begin{equation}\label{QAKubo}
\chi_{J,V} = -\iu \varpi_{\mathrm i}\left([\caI_{\mathrm i}(V),J]\right).
\end{equation}
Moreover, $\chi_{J,V} = \chi^{\mathrm{Kubo}}_{J,V}$.
\end{theorem}
Note that the subscript on $\caI_{\mathrm i}$ indicates that the Hamiltonian involved in the map is $H_{\mathrm i}$. The proof of the theorem is quite simple and we again only sketch it here, see~\cite{OurMathAdiabatic}. We also note that the statement can be extended from local currents to current densities, namely having $J$ being the density of an extensive observable instead of a local observable, see~\cite{TeufelAd}.

On the one hand, Theorem~\ref{thm:MBA} yields
\begin{equation*}
\lim_{\epsilon\to0^+}\rho_{\alpha,\epsilon}(0) = P_{\{\alpha,s=0\}}
\end{equation*}
as $\epsilon\to0^+$. On the other hand 
\begin{equation*}
\partial_\alpha H_{\{\alpha,s=0\}}\vert_{\alpha=0} = V
\end{equation*}
so that the quasi-adiabatic flow yields
\begin{equation*}
\lim_{\alpha\to0^+} \alpha^{-1}\left(P_{\{\alpha,s=0\}} - P_{\{\alpha=0,s=0\}}\right)
= -\iu[\caI_{\mathrm{i}}(V),P_{\{\alpha=0,s=0\}}].
\end{equation*}
Note that both convergences above are uniform in the volume in the weak topology of states, namely when traced against a local observable. This concludes the derivation of (\ref{QAKubo}).

It remains to show equality with Kubo's formula. For this, let $\Lambda$ be arbitrary but fixed. Since for any $A,B\in\caA_\Lambda$, with $Q_{\mathrm{i}}(A) = P_{\mathrm{i}} A (1-P_{\mathrm{i}}) + (1-P_{\mathrm{i}}) A P_{\mathrm{i}}$,
\begin{equation}\label{AntiComm}
\Tr(P_{\mathrm{i}}[A,B]) = \Tr(P_{\mathrm{i}}[Q_{\mathrm{i}}(A),B])
\end{equation}
by cyclicity, we substitute $Q_{\mathrm{i}}(V)$ for $V$ in the commutator of~(\ref{Kubo}), and replace it further by $-\iu[H_{\mathrm{i}}, \caI_{\mathrm{i}}(Q_{\mathrm{i}}(V))]$, see Appendix~\ref{appendix:L}. But then
\begin{equation*}
\chi^{\mathrm{Kubo}}_{J,V} = \lim_{\epsilon\to0^+}\int_0^\infty\ep{-\epsilon \tau}\Tr(P_{\mathrm{i}}[[H_{\mathrm{i}},\ep{-\iu H_{\mathrm i}\tau}  \caI_{\mathrm{i}}(Q_{\mathrm{i}}(V))\ep{\iu H_{\mathrm i}\tau}],J])d\tau,
\end{equation*}
which can be integrated explicitly 
\begin{align}
\int_0^\infty\ep{-\epsilon \tau}[H_{\mathrm{i}},\ep{-\iu H_{\mathrm i}\tau}  \caI_{\mathrm{i}}(Q_{\mathrm{i}}(V))\ep{\iu H_{\mathrm i}\tau}] d\tau
&= -\iu 
\int \frac{\lambda - \mu}{\lambda - \mu+\iu\epsilon} dP(\lambda)\caI_{\mathrm{i}}(Q_{\mathrm{i}}(V))dP(\mu) \nonumber \\
&\quad\stackrel{\epsilon\to0^+}{\longrightarrow}-\iu \caI_{\mathrm{i}}(Q_{\mathrm{i}}(V))\label{JCalc}
\end{align}
since $\vert\lambda - \mu\vert$ is bounded away from zero by the gap assumption. It remains to apply~(\ref{AntiComm}) again to obtain~(\ref{QAKubo}).

Theorem~\ref{thm:LR} assumes a spectral gap, and the corresponding linear response is non-dissipative. Indeed, if work is the change of energy under the driving, then it is easy to see that no work is performed since
\begin{equation*}
\Tr(P_{\mathrm i}[\caI_{\mathrm i}(V), H_{\mathrm i}]) 
= \Tr([H_{\mathrm i},P_{\mathrm i}]\caI_{\mathrm i}(V)) = 0. 
\end{equation*}
The setting is adapted to the quantum Hall effect, where the Lorentz force does similarly perform no work. Since the recent proof~\cite{HastingsMichalakis} of the quantization of conductance for quantum spin systems (see also~\cite{MBQHE} for a short version) assumes as a starting point the validity of linear response, Theorem~\ref{thm:LR} fills the remaining (mathematical) gap in the microscopic understanding of the quantum Hall effect for interacting systems. Note that in the context of the Hall effect where the Hamiltonian is naturally periodic in two variables, the linear response coefficient is usually expressed yet differently, namely as the adiabatic curvature
\begin{equation*}
\iu \Tr(P_{\mathrm i}\:d P_{\mathrm i}\wedge dP_{\mathrm i})
\end{equation*}
of the bundle of ground state projections. The equality of the average of this expression over the torus with Kubo's formula was understood in~\cite{AvronSeiler85}. Another derivation which does not require averaging can be found in~\cite{EquivLR} in this volume. A short history of the problem of quantization of Hall conductance with interactions is in~\cite{YosiIAMP}.

\subsection*{Acknowledgements}

The authors would like to thank Y.~Avron, S.~Teufel and D.~Monaco for discussions.

\appendix
\section{Appendix}

\subsection{A zoo of `local' operators}\label{appendix:Local}

Very briefly, let us recall the general setting of quantum spin systems. It allows for an infinite volume description, although these lecture notes are set on arbitrary but finite volumes. We consider a countable set $\Gamma$ equipped with a metric $d(\cdot,\cdot)$ and a function $F:[0,\infty)\to(0,\infty)$ such that
\begin{equation*}
\Vert F\Vert_1:=\sup_{x\in\Gamma}\sum_{y\in\Gamma}F(d(x,y))<\infty,\qquad
C_F:=\sup_{x,z\in\Gamma}\sum_{y\in\Gamma} \frac{F(d(x,y))F(d(y,z))}{F(d(x,z))}<\infty.
\end{equation*}
Let $\caF_\Gamma$ be the set of finite subsets of $\Gamma$. A finite dimensional complex Hilbert space $\caH_x$ is attached to each $x\in\Gamma$. If $\Lambda\in\caF_\Gamma$, then $\caH_\Lambda := \otimes_{x\in\Lambda}\caH_x$. Let now
\begin{equation*}
\caA_\Lambda :=\caB(\caH_\Lambda),\qquad 
\caA^{\mathrm{loc}} := \cup_{\Lambda\in\caF_\Gamma}\caA_\Lambda,\qquad
\caA := \overline{\caA^{\mathrm{loc}}}^{\Vert\cdot\Vert}
\end{equation*}
be respectively the algebra of observables on $\Lambda$, the algebra of local observables and the algebra of quasi-local observables. In other words, $A\in\caA^{\mathrm{loc}}$ means that there is a finite set $\Lambda$ such that $A\in\caA_\Lambda$, while $A\in\caA$ means that $A$ can be approximated in norm by a sequence of such local observables. The support of $A\in\caA^{\mathrm{loc}}$ is the smallest set $Z$ such that $A\in\caA_Z$. Note that $A\in\caA_\Lambda$ for all $\Lambda\supset \mathrm{supp}(A)$ by tensoring with $1$ on $\Lambda\setminus Z$. 

Now, a local Hamiltonian is a family of operators labelled by $\Lambda\in\caF_\Gamma$ of the form
\begin{equation*}
H_\Lambda = \sum_{Z\subset\Lambda}\Phi(Z),\qquad\Phi(Z)= \Phi(Z)\str\in\caA_Z,
\end{equation*}
such that the interaction potential $\Phi$ has sufficient decay, namely
\begin{equation*}
\Vert\Phi\Vert_F:=\sup_{x,y\in\Gamma}\frac{1}{F(d(x,y))}\sum_{Z\ni \{x,y\}}\Vert \Phi(Z)\Vert<\infty
\end{equation*}
Although this is meaningful (and defines a norm) for any $F$ introduced above, the adiabatic theorem requires $F$ to decay faster than any inverse power, in the sense that $\sup\{(1+r^k) F(r):r\in[0,\infty)\}<\infty$ for all $k\in\bbN$.

When we say that a local Hamiltonian is of order $\epsilon$, we mean that it can be associated with a potential whose $F$-norm is of that order. The use of this local norm is essential since a local Hamiltonian is not a local operator, but rather an extensive one since
\begin{equation*}
\Vert H_\Lambda \Vert \leq \vert \Lambda\vert \Vert F\Vert_1 \Vert \Phi\Vert_F.
\end{equation*}
Its locality becomes apparent through the following commutator estimate which was used many times in these notes,
\begin{equation*}
\Vert [H_\Lambda,A]\Vert \leq 2\Vert A\Vert \vert \mathrm{supp}(A)\vert \Vert F\Vert_1 \Vert\Phi\Vert_F
\end{equation*}
where the bound depends in particular on the support of $A$, but not on $ \Lambda $.

Finally, this can be bootstrapped to yield the Lieb-Robinson bound for the dynamics generated by a local Hamiltonian. For any local observables $A,B\in\caA^{\mathrm{loc}}$ with disjoint supports separated by $d>0$, there is $v>0$ such that
\begin{equation*}
\left\Vert[\exp(\iu t H_\Lambda) A \exp(-\iu t H_\Lambda),B]\right\Vert\leq C(A,B)\exp(-\mu(d - v\vert t\vert))
\end{equation*}
for a $\mu>0$, and the bound is independent of $\Lambda$. As a consequence, local observables are mapped into almost local ones and more generally, the algebra of almost local observables is invariant under the dynamics, see \eg Appendix~C of~\cite{bachmann2016lieb}. An observable $A$ is called almost local if there is a $Z\in\caF_\Gamma$ and a sequence $A_n\in\caA_{Z^n}$ (where $Z_n = \{x\in\Gamma:d(x,Z)\leq n\}$) such that
\begin{equation*}
\Vert A - A_n \Vert \leq C_k \Vert A \Vert \mathrm{supp}(A) n^{-k}
\end{equation*}
for all $k\in\bbN$. In other words, almost local observables are quasi-local observables whose finite volume approximations converge very rapidly.

\subsection{On the inverse of $-\iu[H,\cdot]$}\label{appendix:L}

Let $H$ be a self-adjoint operator on a finite dimensional Hilbert space $\caH$ whose spectrum lies within two disjoint intervals $I_1, I_2$ with $\mathrm{dist}(I_1, I_2)=\gamma>0$. The rank-one operators $\vert \psi_k\rangle\langle\psi_l\vert$, where $\{\psi_j:j\in\{1,\ldots,\mathrm{dim}(\caH)\}\}$ is an eigenbasis of $H$, are eigenvectors of
\begin{equation*}
L = -\iu[H,\cdot]
\end{equation*}
for the eigenvalue $-\iu(\lambda_k - \lambda_l)$. 

Let $P = \chi_{I_1}(H)$ be the spectral projection of $H$ associated with $I_1$, and let $Q$ be the Banach space projection
\begin{equation*}
Q(A) = (1-P) A P + P A (1-P)
\end{equation*}
onto off-diagonal matrices with respect to the orthogonal decomposition
\begin{equation*}
\caH = P\caH\oplus (1-P)\caH.
\end{equation*}
Since $L(\mathrm{Ran}(Q))\subset\mathrm{Ran}(Q)$, the restriction $L\upharpoonright_{\mathrm{Ran}(Q)}$ is well-defined, its spectrum is
\begin{equation*}
\mathrm{Spec}(\iu L\upharpoonright_{\mathrm{Ran}(Q)})\subset\bbR\setminus(-\gamma,\gamma),
\end{equation*}
and it is invertible with bounded inverse. Since $H$ has a gap, the limit
\begin{equation}\label{J}
\caJ(A):=\lim_{\eta\to 0^+}\int_{0}^{\infty}\ep{-\eta t}\ep{\iu t H} A \ep{-\iu t H} dt
\end{equation}
exists for all $A\in \mathrm{Ran}(Q)$ and satisfies
\begin{equation}\label{LInverse}
A = -\iu[H,\caJ(A)],
\end{equation}
see~(\ref{JCalc}), with $H\to-H$.

If $\caH$ is now the Hilbert space of a quantum spin system and $H$ is a local Hamiltonian, then it generates a local dynamics by the Lieb-Robinson bound. It is therefore natural to look for an inverse that does so, too. However, since the integral in~(\ref{J}) extends to infinity as the cutoff is removed $\eta\to 0^+$, $\caJ(A)$ is in general supported on the whole system even if $A$ was strictly local. 

With $W$ as in Section~\ref{sec:Locality}, 
\begin{equation}\label{I}
\caI(A):=\int_{-\infty}^{\infty}W(t)\ep{\iu t H} A \ep{-\iu t H} dt
\end{equation}
also satisfies~(\ref{LInverse}) on $\mathrm{Ran}(Q)$ by~(\ref{FourierW}). By the Lieb-Robinson bound for small $t$ and the fast decay of $W$ at infinity, the algebra of almost local observables is invariant under $\caI$. It further follows that if $H$ is a local Hamiltonian, then so is $\caI(H)$, see Lemma~4.8 of~\cite{OurMathAdiabatic}. Finally,

\begin{lemma}\label{lma:LocalInverse}
For any $A\in\caB(\caH)$,
\begin{equation*}
[A,P] = -\iu\big[[H,\caI(A)],P\big].
\end{equation*}
\end{lemma}
\begin{proof}
By Jacobi's identity, this follows immediately from~(\ref{LInverse}) with $\caI$ instead of $\caJ$, since $[A,P] = (1-P)AP - PA(1-P)\in\mathrm{Ran}(Q)$.
\end{proof}

\bibliography{LocAd}
\bibliographystyle{amsplain}

\end{document}